\documentclass[12pt]{article}
\usepackage[utf8]{inputenc}

\usepackage{amsfonts}
\usepackage{amsmath}
\usepackage{amsthm}
\usepackage{mathtools}
\usepackage{dsfont}
\usepackage{mathrsfs}
\usepackage{MnSymbol}
\usepackage{wasysym}
\usepackage{authblk}
\usepackage{csquotes}
\usepackage{xcolor}
\usepackage[T1]{fontenc}
\usepackage{hyperref}
\hypersetup{
    colorlinks=true,
    linkcolor=blue,
    filecolor=magenta,      
    urlcolor=cyan,
    pdftitle={Overleaf Example},
    pdfpagemode=FullScreen,
    }

\usepackage{appendix}
\usepackage{geometry}
\usepackage{tikz-cd}
\usetikzlibrary{positioning}
\usetikzlibrary{matrix,arrows}
\usetikzlibrary{calc}
\usetikzlibrary{shapes, shapes.geometric, shapes.symbols, shapes.arrows, shapes.multipart, shapes.callouts, shapes.misc,decorations.pathmorphing}
\tikzset{snake it/.style={decorate, decoration=snake}}
\usepackage{tikz-3dplot}
\usepackage{dsfont}

\newtheorem{theorem}{Theorem}
\newtheorem{definition}{Definition}

\newtheorem{proposition}{Proposition}

\newcommand{\af}{\alpha}

\DeclareMathOperator{\tr}{Tr}

\newcommand{\ketbra}[2]{\vert #1 \rangle \langle #2 \vert }
\newcommand{\sketbra}[1]{\vert #1 \rangle \langle #1 \vert }

\title{Kochen and Specker's view on functional relations conflicts with the collapse postulate
}
\author{Alisson Tezzin\footnote{alisson.tezzin@usp.br}}
\affil{\small\textit{Department of Mathematical Physics, Institute of Physics,
University of São Paulo
\\
R. do Matão 1371, São Paulo 05508-090, SP, Brazil}}
\date{\today}

\begin{document}
\maketitle
\begin{abstract}
    A key ingredient of the Kochen-Specker theorem is the so-called functional composition principle, which asserts that hidden states must ascribe values to observables in a way that is consistent with all functional relations between them. This principle is motivated by the assumption that, like functions of observables in classical mechanics, a function $g(A)$ of an observable $A$ in quantum theory is simply a logically possible observable derived from $A$, and that measuring $g(A)$ consists in measuring $A$ and post-processing the resulting value via $g$. In this paper we show that, under suitable conditions, this reasonable assumption is in conflict with the collapse postulate. We then discuss possible solutions to this conflict, and we argue that the most justifiable and less radical one consists in adapting the collapse postulate by taking measurement contexts into account. This dependence on contexts arises for precisely the same reason why noncontextual hidden-variable models are ruled out by the Kochen-Specker theorem, namely that an observable $A$ can typically be written as a function $A=g(B)=h(C)$ of noncommuting observables $B,C$.
    
\end{abstract}

\section{Introduction}\label{sec: introduction}

In classical mechanics, the \textit{pure state} of a  system of $N$ particles 
is completely specified by a pair $(q,p)$, where $q \in \mathbb{R}^{3N}$ encodes the positions of all particles with respect to a previously fixed reference frame and $p \in \mathbb{R}^{3N}$ encodes their momenta \cite{arnol2013mathematical}. This state is a truth-maker, which means that, as Chris Isham and Andreas Döring put it, it specifies ``the way things are at a particular moment in time'' \cite{doring2010thing, IshamDoeringI}. This is why it is assumed that any \textit{observable} $\underline{A}$ of the system  can be represented as a Borel function $A$ on its state  space $\Lambda$:\footnote{For the sake of simplicity, we consider $\Lambda = \mathbb{R}^{3N} \times \mathbb{R}^{3N}$} this function simply assigns to each pure state $\lambda  \equiv (q,p) \in \Lambda$ the value $\underline{A}$ assumes when the pure state of the system is $\lambda$ \cite{doring2010thing, IshamDoeringI, KS1967}.
Obvious examples of relevant observables in classical mechanics are position ($\underline{Q}$) and momentum ($\underline{P}$), which are represented by the functions $Q(q,p) \doteq q$ and $P(q,p) \doteq p$ respectively --- or, to restrict the discussion to real-valued functions, $Q^{(k)}_{i}$ and $P^{(k)}_{i}$, where $Q^{(k)}_{i}$ returns the component $i\in \{1,2,3\}$ of the position of the $k$-th particle, and analogously for $P^{(k)}_{i}$. All other observables are, in the last instance, theoretical constructs which are convenient for a proper description of the system, like ``total energy'' ($\underline{E}$), ``potential energy'' ($\underline{E_{P}}$) or ``angular momentum'' ($\underline{L}$), and many of these observables are derived as functions of others, like ``square of total energy'' ($\underline{E^{2}}$). It means that there are \textit{functional relations} between observables in classical mechanics, and the fact that these observables are constructed within the theory guarantees that these functional relations coexists in harmony with the theoretical framework. To be precise, let $A_{1},\dots,A_{n}: \Lambda \rightarrow \mathbb{R}$ ($n=1$ is a particular case) be measurable functions representing observables $\underline{A_{1}}, \dots, \underline{A_{n}}$ respectively, and let $g: \mathbb{R}^{n} \rightarrow \mathbb{R}$ be any Borel function. We ``define'' an observable $\underline{g(A_{1},\dots,A_{n})}$ as the observable associated with the function $g \circ (A_{1},\dots,A_{n})$. Note that, by construction, any observable in classical mechanics is a function of $\underline{P}$ and $\underline{Q}$, given that any  function $A$ on $\Lambda$ satisfies $A(q,p) = (A \circ (Q,P))(q,p)$. Now let $V_{\lambda}$ be the valuation defined by a state $\lambda \in \Lambda$, that is to say, $V_{\lambda}(H) \doteq H(\lambda)$ for any Borel function $H$. Then
\begin{align}
    V_{\lambda}(g(A_{1},\dots,A_{n})) =\label{eq: first functional relation} g(V_{\lambda}(A_{1},\dots,A_{n}))
\end{align}
for every Borel functions $A_{1},\dots,A_{n}:\Lambda \rightarrow\mathbb{R}$ and $g: \mathbb{R}^{n} \rightarrow\mathbb{R}$. In particular, given $A,B: \Lambda \rightarrow \mathbb{R}$, $g:\mathbb{R} \rightarrow \mathbb{R}$ and $\af,\beta \in \mathbb{R}$, 
\begin{align}
    V_{\lambda} (\af A + \beta B) &=\label{eq: linearity} \af V_{\lambda}(A) + \beta V_{\lambda}(B)
    \\
    V_{\lambda}(AB) &=\label{eq: homomorphism} V_{\lambda}(A)V_{\lambda}(B)
    \\
    V_{\lambda}(g(A)) &=\label{eq: functional relations} g(V_{\lambda}(A)).
\end{align}
It is also important to note that the value $g(V_{\lambda}(A))$ of $H \doteq g(A)$ is actually independent from the functional relation between $A$ and $H$, which means that, if we have $g(A) = H = g(B)$, then, for any state $\lambda$, $g(V_{\lambda}(A)) = V_{\lambda}(H) = g(V_{\lambda}(A))$. This is a straightforward consequence of the fact that $H$ itself is a function on $\Lambda$, and thus has a definite value in each state $\lambda$. Hence, since the pure state of a physical system in classical mechanics is well defined at any given time, all observables  have definite values at all times, and these values are compatible with all functional relations between them.



In the search for ``hidden variables'' completing the description of microscopic systems provided by quantum theory, valuations and functional relations between observables have played a major role \cite{doring2010thing, KS1967, isham1998topos}. In fact, if a hidden variable is like a classical pure state, it must ascribe a  definite value to each observable of the system. Thus, in a quantum system  represented by a Hilbert space $H$, if all bounded selfadjoint operators in $H$ represent observables  of this system, a hidden state will define a mapping $V: \mathcal{B}(H)_{\text{sa}} \rightarrow \mathbb{R}$, where $\mathcal{B}(H)_{\textit{sa}}$ is the collection of all bounded selfadjoint operators on $H$.  $\mathcal{B}(H)_{\text{sa}}$ is a real vector space, and it is embedded in the C*-algebra $\mathcal{B}(H)$ of all bounded operators, so, like the collection of all Borel functions on $\Lambda$, it is endowed with an algebraic structure. It thus seems reasonable, at first glance, to assume that the valuation $V$ defined by a hidden state must be compatible with this vector space structure, that is to say, that it must satisfy equation \ref{eq: linearity}. Notwithstanding von Neumann had proved that, under suitable conditions, no such valuation can exist 
\cite{landsman2017foundations}, it has been pointed out by many authors, specially John Bell \cite{norsen2017foundations, MerminBell}, that this linearity assumption is unfounded, since it conflicts with quantum predictions when noncommuting observables are taken into account \cite{norsen2017foundations}. In face of the restrictions imposed by noncommuting operators, assuming that $V$ simply preserves functional relations between single observables, namely  that it satisfies equation  \ref{eq: functional relations}, seems to be a reasonable alternative. 
To begin with, given an observable $A$ and a function $g(A) \in \mathcal{B}(H)_{\text{sa}}$ of $A$, where $g$ is a real Borel function on the spectrum of $A$ and $g(A)$ is defined according to the Borel functional calculus, $g(A)$ and $A$ commute \cite{barata2006curso, kadison1997fundamentals}. Furthermore, since commuting observables $A,B$ can always be written as functions of a third observable $C$ \cite{isham1998topos, landsman2017foundations}, namely $A = g(C)$ and $B = h(C)$ for real Borel functions $g,h$ on $\sigma(C)$, and since the Borel functional calculus is a homomorphism \cite{barata2006curso, kadison1997fundamentals}, any valuation preserving functional relations between single observables is immediately quasi-linear \cite{landsman2017foundations}, which means that it satisfies equation \ref{eq: linearity} whenever $A$ and $B$ commutes, and it also satisfies equation \ref{eq: homomorphism} for commuting observables. It justifies the following definition. 

\begin{definition}[Valuation, \cite{isham1998topos, doring2005kochen}]\label{def: valuation} Let $H$ be a separable Hilbert space, and let $\mathcal{B}(H)_{\text{sa}}$ be the collection of all bounded selfadjoint operators in $H$. A function $V: \mathcal{B}(H)_{\text{sa}} \rightarrow \mathbb{R}$ is said to be a valuation if it satisfies the following conditions.
\begin{itemize}
    \item[(a)] The `value rule' is satisfied, i.e., for any operator $A \in \mathcal{B}(H)_{\text{sa}}$,
    \begin{align}
        V(A) \in \sigma(A),
    \end{align}
    where $\sigma(A)$ denotes the spectrum of $A$.
    \item[(b)]  The `functional composition principle is satisfied'. That is, if $B = g(A)$ for some Borel function $g: \sigma(A) \rightarrow \mathbb{R}$, 
    \begin{align}
        V(g(A)) = g(V(A)).
    \end{align}
\end{itemize}
\end{definition}

Kochen-Specker theorem \cite{KS1967}, however, implies that valuations do not exist if $\text{dim}(H) > 2$ \cite{KS1967, isham1998topos, doring2005kochen}.

The functional composition principle is important not only because it guarantees that linear combinations and products of commuting observables are preserved, but also --- and mainly --- because it is commonly assumed that functional relations between single observables in quantum theory  are analogous to functional relations in classical mechanics. Kochen and Specker themselves make it explicit in their seminal paper:

\begin{displayquote}
``Now it is clear that the observables of a theory are in fact not independent. The observable $A^{2}$ is a function of the observable $A$ and is certainly not independent of $A$. \textbf{In any theory, one way of measuring $A^{2}$ consists in measuring $A$ and squaring the resulting value. In fact, this may be used as a \textit{definition} of a function of an observable}. (...) This definition coincides with the definition of a function of an observable \textbf{in both quantum and classical mechanics.}'' (\cite{KS1967}, emphasis added).
\end{displayquote}
In the excerpt we suppressed from the quotation, Kochen and Specker precisely explain what they mean by ``a \textit{definition} of a function of observable''. In order to understand their definition, we need to go into the details of their formalism, so we turn our attention to it now. To begin with, they provide a sketch of a physical theory, which consists basically of a triple $(\mathcal{O},\mathcal{S},P)$, where $\mathcal{O}$ and $\mathcal{S}$ are nonempty sets whose elements represent, respectively, observables and states, and $P$ is a mapping which assigns, to each pair $(\rho, A) \in \mathcal{S} \times \mathcal{O}$, a probability measure $P_{\rho}( \ \cdot \ ,  A)$ on $\mathbb{R}$ (endowed with an appropriate $\sigma$-algebra). According to this prototype physical theory, given any set $\Delta \subset \mathbb{R}$ which is measurable with respect to $P_{\rho}(\ \cdot \ ,A)$, $P_{\rho}(\Delta,A)$ denotes the probability that, for a system in the state $\rho$, the measurement of $A$ yields a value lying in $\Delta$ \cite{KS1967}.

In the particular case of classical mechanics, where observables are measurable functions on a space of pure states $\Lambda$, a (not necessarily pure) state consists in a probability measure $\mu$ on $\Lambda$, and, for any measurable set $U \subset \Lambda$, $\mu(U)$  is the probability that the pure state of the system lies in $U$ \cite{KS1967} --- consequently, pure states define Dirac measures. Therefore, given any state $\mu$ and any observable $A$, $P_{\mu}( \ \cdot \ ,A)$ is simply the pushforward measure $A_{\ast}(\mu)$, which means that
\begin{align}\label{eq: pushforward}
    P_{\mu}(\Delta, A) = \mu(A^{-1}(\Delta)).
\end{align}
In quantum theory, observables of a given physical system are represented by bounded selfadjoint operators in a separable Hilbert space $H$, and states are normalized positive linear functionals on the C*-algebra $\mathcal{B}(H)$ of bounded operators.\footnote{Kochen and Specker do not restrict the definition of observable to bounded operators, as we do. Also, they take only pure states into account.} Note that, if $\text{dim}(H) < \infty$, there is a one-to-one correspondence between normalized positive linear functionals on $\mathcal{B}(H)$ and density operators, which is established by the mapping $\rho \mapsto \tr(\rho \ \cdot \ )$. Given any state $\omega$, and any observable $A$, the probability measure $P_{\omega}( \ \cdot \ , A)$ is defined by the Born rule, which means that
\begin{align}
    P_{\omega}(\Delta,A) = \omega(\chi_{\Delta}(A)),
\end{align}
where $\chi_{\Delta}(A)$ is the orthogonal projection corresponding to the subset $\Delta \cap \sigma(A)$ of $\sigma(A)$, which is defined according to the Borel functional calculus \cite{landsman2017foundations, barata2006curso, kadison1997fundamentals}. 

We can now reveal the definition we hid in the previous quotation. It goes as follows:
\begin{displayquote}
``\textbf{We define the observable $\mathbf{g(A)}$} for every observable $A$ and Borel function $g:\mathbb{R} \rightarrow \mathbb{R}$ \textbf{by the formula}
\begin{align}\label{eq: KS definition}
    P_{\rho}(\Delta, g(A)) = P_{\rho}(g^{-1}(\Delta),A)
\end{align}
for each state $\rho$. If we assume that every observable is determined by the function $P$, i.e., $P_{\rho}( \ \cdot \ ,A ) = P_{\rho}( \ \cdot \ ,B)$ for every state $\rho$ implies that $A=B$, then the formula \ref{eq: KS definition} defines the observable $g(A)$.  \textbf{This definition coincides with the definition of a function of an observable in both quantum and classical mechanics}.'' (\cite{KS1967}, emphasis added).
\end{displayquote}
It is important to mention that Kochen and Specker write $P_{g(A)\rho}$ and $P_{A\rho}$ rather than $P_{\rho}( \ \cdot \ ,g(A))$ and  $P_{\rho}( \ \cdot \ ,A)$. In any case, they conclude as follows:
\begin{displayquote}
``Thus \textbf{the measurement of a function $\mathbf{g(A)}$ of an observable $\mathbf{A}$ is independent of the theory considered --- one merely writes $\mathbf{g(\boldsymbol{\af})}$ for the value of $\mathbf{g(A)}$ if $\mathbf{\boldsymbol{\af}}$ is the measured value of $\mathbf{A}$.} The set of observables of a theory thereby acquires an algebraic structure, and the introduction of hidden variables into a theory should preserve this structure.'' (\cite{KS1967}, emphasis added).
\end{displayquote}

We can thus say that, according to Kochen and Specker's view on functional relations, an  observable $g(A)$ is the theoretical representation of a \textit{logically possible observable} derived from $A$, and that measuring $g(A)$ consists in measuring $A$ and  post-processing the resulting value via $g$ (i.e., evaluating $g$ on it). From now on, this is what we mean when we say that $g(A)$ is ``a logically possible observable representing an experimental post-processing of $A$ via $g$'', or simply ``a post-processing of $A$ via $g$''. 

In this paper we show that, as reasonable as this view on functional relations may be, it conflicts with the collapse postulate in the particular case of quantum theory. This is shown in section \ref{sec: conflict}. In the next section, we discuss how to incorporate the idea of state collapse into Kochen and Specker's framework, and in section \ref{sec: collapse} we examine the collapse postulate in quantum theory, in order to understand how the experimentalist updates the state of the system after obtaining a set $\Delta \subset \sigma(A)$ of possible outcomes in a measurement of $A$. In section \ref{sec: discussion}, we analyze possible solutions to the conflict between functional relations and the collapse postulate, and we argue that the most reasonable way of avoiding it consists in incorporating the measurement context --- which is no more than the basis the experimentalist chooses to perform the measurement --- into the collapse postulate. According to the definition we propose, the collapse due to a measurement of an observable $A$ depends on the measurement context if and only if $A$ can be written as a function of noncommuting observables, namely $A = g(B) = h(C)$, where $[B,C] \neq 0$. The existence of such observables  is precisely the reason why valuations on $\mathcal{B}(H)_{\text{sa}}$ (and consequently noncontextual hidden-variable models for $H$) do not exist \cite{isham1998topos, doring2005kochen}, so the dependence on contexts we come up with is in accordance with Kochen-Specker theorem. 

\section{The collapse postulate in  Kochen and Specker's framework}\label{sec: general collapse}

In Ref.~\cite{KS1967}, Kochen and Specker restrict their analysis to the predictions a theory provides, as their description of a physical theory makes clear. From this perspective, the fact that an observable $g(A)$ satisfies equation \ref{eq: KS definition} for every state $\rho$ and every measurable set $\Delta$ justifies their claim that $g(A)$ can be seen as the theoretical representation of a post-processing of $A$ via $g$, since the probability measure $P_{\rho}(g^{-1}( \ \cdot \ ),A)$ that appears at the right hand side of equation \ref{eq: KS definition} is precisely the probability measure describing this post-processing, i.e., for any measurable set $\Delta$, $P_{\rho}(g^{-1}(\Delta),A)$ is the probability that, for a system in the state $\rho$, an experimental post-processing of $A$ via $g$ yields a value lying in $\Delta$. However, their description of a physical theory captures only, say, half of quantum theory, since it includes the Born rule but excludes the collapse postulate. We can thus ``complete'' their framework as follows:

\begin{definition}[$\mathfrak{F}$-system]\label{def: F-system} Let $\mathfrak{F}$ be a physical theory. A  $\mathfrak{F}$-system consists in a quadruple $(\mathcal{O},\mathcal{S},P,T)$, where $\mathcal{O}$ and $\mathcal{S}$ are nonempty sets, and $P$ and $T$ are defined as follows.
\begin{itemize}
    \item[(a)] $P$ is a mapping which associates, to each pair $(\rho,A) \in \mathcal{S} \times \mathcal{O}$, a Borel measure $P_{\rho}( \ \cdot \ ,A)$ on $\mathbb{R}$. For any Borel set $\Delta$,  $P_{\rho}( \Delta,A)$ is the probability that, for a system in the state $\rho$, a measurement of $A$ yields a value lying in $\Delta$ \cite{KS1967}.
    \item[(b)] $T$ is a mapping which associates, to each \textit{measurement event} $(\Delta,A) \in \mathfrak{B}(\mathbb{R}) \times \mathcal{O}$, a mapping $T_{(\Delta,A)}: \mathcal{S} \rightarrow \mathcal{S}$, where $\mathfrak{B}(\mathbb{R})$ denotes the Borel $\sigma$-algebra on $\mathbb{R}$. For a system in the state $\rho$, $T_{(\Delta,A)}(\rho)$ represents its updated state immediately after the measurement event $(\Delta,A)$, that is to say, after the experimentalist extracts the information that some outcome in $\Delta$ (unknown to her) has been obtained in a measurement of $A$.
\end{itemize}
\end{definition}
Note that, for the sake of simplicity, we assume that $P_{\rho}( \ \cdot \ ,A)$ is always a Borel measure, and that $\Delta$ is always a Borel set. We denote by $\mathbb{E}$ the collection $\mathfrak{B}(\mathbb{R}) \times \mathcal{O}$ of all \textit{measurement events}, and each state $\rho$ trivially defines a function $P_{\rho}: \mathbb{E} \rightarrow [0,1]$. We say that $P_{\rho}(\Delta,A)$ is ``the probability of the event $(\Delta,A)$ for a system in the state $\rho$'', and we say that the measurement event $(\Delta,A)$ ``has happened'' or ``has occur'' to indicate that a measurement of $A$ has been performed and some outcome lying in $\Delta$ has been obtained.

When a measurement of $A$ is performed, an outcome $\af \in \mathbb{R}$ is obtained, whether or not the experimentalist has access to it. Therefore, saying that a measurement event $(\Delta,A)$ has happened is  equivalent to say that one, and only one, of the events $\{(\{\af\},A): \af \in \Delta\}$ has occurred. For this reason, we will distinguish between the \textit{objective event} $(\af,A) \equiv (\{\af\},A)$ and the \textit{subjective event} $(\Delta,A)$ that occur when an observable $A$ is measured; the  former is determined by the outcome (usually unknown) that is obtained, and the later encodes how much information the experimentalist extracts from the system by measuring $A$. Note that, in classical mechanics, we are usually dealing with subjective events, since we have experimental uncertainly --- $\Delta$ is usually an interval $(\af - \delta_{A}, \af + \delta_{A})$, where $\delta_{A}>0$ is the experimental error associated with the experimental apparatus. This is the case with infinite-dimensional quantum systems either, given that, in infinite-dimensional Hilbert spaces, some observables, like position and momentum, have no eigenvalues at all. Finally, note that, in definition \ref{def: F-system}, it is the subjective  event that dictates how the state of the system has to be updated. Thus, depending on the theory $\mathfrak{F}$, the update may be purely subjective, as in classical mechanics, or it may be a ``mixture'' of an objective collapse, namely the collapse determined by the objective event, and an update on the knowledge of the experimentalist about the pure state of the system; this is the case of quantum theory.

If Kochen and Specker are right in saying that equation \ref{eq: KS definition} defines $g(A)$, then, for any Borel set $\Delta$, the measurement event $(\Delta,g(A))$ must update the state of the system in precisely the same way as $(g^{-1}(\Delta),A)$, since, according to their definition, these measurement events are equivalent in every possible sense: measuring $g(A)$ consists in measuring $A$ and post-processing the resulting value via $g$, thus saying that ``an outcome in $\Delta$ has been obtained in a measurement of $g(A)$'', that is to say, saying that measurement event $(\Delta,g(A))$ has occurred, is just another way of saying that ``an outcome in $g^{-1}(\Delta)$ has been obtained in a measurement of $A$'', which means that the event $(g^{-1}(\Delta),A)$ has happened. It motivates the following definition.
\begin{definition}[Post-processing]\label{def: post-processing}
Let $\mathfrak{F}$ be a physical theory, and let $(\mathcal{O},\mathcal{S},P,T)$ be a $\mathfrak{F}$-system. Let $A$ and $B$ be observables in this system, and let $g: \mathbb{R} \rightarrow \mathbb{R}$ be a Borel function. We say that $B$ is a ``logically possible observable representing an experimental post-processing of $A$ via $g$'', or simply a ``post-processing of $A$ via $g$'', iff the following conditions are satisfied.
\begin{itemize}
    \item[(a)] For any state $\rho$, the probability measure defined by $B$ matches the probability measure defined by an experimental post-processing of $A$ via $g$, i.e.,
    \begin{align}
        P_{\rho}( \ \cdot \ , B) = P_{\rho}(g^{-1}( \ \cdot \ ), A).
    \end{align}
    It means that equation \ref{eq: KS definition} is satisfied.
    \item[(b)] For any Borel set $\Delta$, the events $(\Delta,B)$ and $(g^{-1}(\Delta),A)$ update the state of the system in the same way, that is to say,
    \begin{align}
        T_{(\Delta,B)} = T_{(g^{-1}(\Delta),A)}.
    \end{align}
\end{itemize}
\end{definition}
Let $\mathfrak{F}$ be classical mechanics, and let $\Lambda$ be the space of pure states of a system $(\mathcal{O},\mathcal{S},P,T)$; for the sake of simplicity, let's assume that $\Lambda = \mathbb{R}^{n}$. As we mentioned, a (not necessarily pure) state of the system consists in a Borel measure $\mu$ on  $\Lambda$, and, for any Borel set $U \subset \Lambda$, $\mu(U)$ is the probability that the pure state of the system lies in $U$ \cite{KS1967}. This state can be understood as the degree of knowledge of the experimentalist about the pure state of the system. 
After an event $(\Delta,A)$, her  knowledge is updated, and she ends up with the conditional probability
\begin{align}
    \mu( \ \cdot \ \vert A^{-1}(\Delta)) \doteq\label{eq: classical update} \frac{\mu( \ \cdot \ \cap A^{-1}(\Delta))}{\mu(A^{-1}(\Delta))}.
\end{align}
This is the state of the system immediately after the measurement event $(\Delta,A)$, thus, if $\mu(A^{-1}(\Delta)) \neq 0$,
\begin{align}
    T_{(\Delta,A)}(\mu) &\doteq \mu( \ \cdot \ \vert A^{-1}(\Delta)).
\end{align}
If $\mu(A^{-1}(\Delta)) = 0$, we define $T_{(\Delta,A)}(\mu)$ as the null measure $0$. Note that the null measure is not a probability measure, but, for practical purposes, it is useful to include it in the set of states, given that we would otherwise be forced to restrict the domain of $T_{(\Delta,A)}$ to the collection of states satisfying $T_{(\Delta,A)}(\mu) \neq 0$. In any case, the state $\mu$ of a system will never be updated to the null measure, since  $T_{(\Delta,A)}(\mu) = 0$ if and only if $P_{\mu}(\Delta,A) = 0$, which means that any event satisfying $T_{(\Delta,A)}(\mu) = 0$ is impossible for a system in the state $\mu$. 

The mapping $(\Delta,A) \mapsto A^{-1}(\Delta)$ canonically induces an equivalence relation $\sim$ in the collection $\mathbb{E}$ of all measurement events, namely $(\Delta,A) \sim (\widetilde{\Delta},B)$ iff $A^{-1}(\Delta) = B^{-1}(\widetilde{\Delta})$. Both the probability of an event $(\Delta,A)$ and the way this event updates the state of the system depend solely on the measurable set $A^{-1}(\Delta)$ associated with it, thus equivalent events affects the system in precisely the same way. Consequently, if an observable $B$ is a function $g(A)$ of another observable $A$, i.e., if $B = g \circ A$ for some Borel function $g: \mathbb{R} \rightarrow \mathbb{R}$, $B$ is a logically possible observable representing an experimental post-processing of $A$ via $g$, in the sense of definition \ref{def: post-processing}, given that, for any Borel set $\Delta$, $B^{-1}(\Delta) = (g \circ A)^{-1}(\Delta) = A^{-1}(g^{-1}(\Delta))$, which means that $(\Delta,B)$ and $(g^{-1}(\Delta),A)$ are equivalent. As we show next, things are not so simple in quantum theory.

\section{The collapse postulate in quantum theory}\label{sec: collapse}
For the sake of simplicity, we will focus on finite-dimensional quantum systems from now on. In this case, we can identify a quantum system $(\mathcal{O},\mathcal{S},P,T)$ (see definition \ref{def: F-system}) with a finite-dimensional Hilbert space $H$.  Any selfadjoint operator in $H$ is assumed to be a valid observable, thus $\mathcal{O} = \mathcal{B}(H)_{\text{sa}}$, and for the same reason why we considered the null measure as a state in classical mechanics, we will include the null operator in $\mathcal{S}$, hence $\mathcal{S} = \mathcal{D}(H) \cup \{0\}$, where $\mathcal{D}(H)$ is the collection of all density operators in $H$ --- despite that, unless explicitly stated otherwise, by a state we mean a density operator. As we have seen, the mapping $P$ is defined by the Born rule, which means that, given any state $\rho$ and any observable $A$, we have, for any Borel set $\Delta$, \cite{KS1967, landsman2017foundations}
\begin{align}
    P_{\rho}(\Delta,A) \doteq \tr(\rho E_{\Delta}),
\end{align}
where $E_{\Delta} \equiv \chi_{\Delta}(A)$. It is worth to emphasize that, in finite-dimensional Hilbert spaces,
\begin{align}
    E_{\Delta} &= \sum_{\af \in \sigma(A) \cap \Delta} E_{\af},
\end{align}
where $E_{\af} = \chi_{\{\af\}}(A)$ is the projection onto the subspace spanned by the eigenvalue $\af$ of $A$. Since $\sigma(A)$ is a finite set, it is useful to consider the probability distribution $p^{\rho}_{A}$ induced by $P_{\rho}( \ \cdot \ ,A)$ on $\sigma(A)$, namely \cite{landsman2017foundations}
\begin{align}
    p^{\rho}_{A}(\af) &\doteq P_{\rho}(\{\af\},A) = \tr(\rho E_{\af}).
\end{align}

For practical purposes, we will assume from now on  that $\Delta \subset \sigma(A)$. If $\Delta$ is a singleton $\{\af\}$, we will write $(\af,A)$ rather than $(\{\af\},A)$. 

If the experimentalist has access to the outcome of a measurement, she can use the collapse postulate to apprehend the state of the system immediately after this procedure:

\begin{definition}[Collapse postulate, \cite{hannabuss1997introduction, nielsen_chuang_2010, hall2013quantum}]\label{def: collapse postulate} Let $H$ be a finite-dimensional Hilbert space. When a measurement event $(\af,A)$ occurs, that is to say, when a measurement of $A$ yields the outcome $\af \in \sigma(A)$, the state $\rho$ of the system is updated to 
\begin{align}
    \rho^{A}_{\af} \doteq \frac{E_{\af} \rho E_{\af}}{\tr(E_{\af} \rho E_{\af})} =  \frac{E_{\af} \rho E_{\af}}{p^{\rho}_{A}(\af)},
\end{align}
where $E_{\af} \equiv \chi_{\{\af\}}(A)$.
\end{definition}
Hence, for any event $(\af,A)$, $\af \in \sigma(A)$, the mapping $T_{(\af,A)}: \mathcal{S}\rightarrow \mathcal{S}$ (see definition \ref{def: F-system}) is given by
\begin{align}
    T_{(\af,A)}(\rho) \doteq
    \begin{cases}
    \rho^{A}_{\af} \ \qquad \text{if} \ p^{\rho}_{A}(\af) \neq 0
    \\
    0  \ \ \qquad \text{otherwise}.
    \end{cases}
\end{align}


Now, consider a subjective measurement event $(\Delta,A)$, and, for the sake of argument, assume that $\Delta = \{\af,\af'\}$, where $\af \neq \af'$. As we have mentioned in section \ref{sec: general collapse}, if all the experimentalist knows is that $(\Delta,A)$ has occurred, then the information she has about the objective event that has taken place is that it is  $(\af,A)$ or $(\af',A)$. If she believes in the collapse postulate, she will conclude that the state of the system immediately after this event is necessarily $\rho^{A}_{\af}$ or $\rho^{A}_{\af'}$. Furthermore, she knows what is the probability of each one of these alternatives. In fact, the state of the system  is $\rho^{A}_{\af}$ if and only if $\af$ has been obtained, and, under the evidence that $\Delta$ has occurred, the probability of $\af$ (or, equivalently, $(\af,A)$) is given by the conditional probability
\begin{align}
    P^{\rho}_{A}(\{\af\} \vert \Delta) &= \frac{P^{\rho}_{A}(\{\af\} \cap \Delta  )}{P^{\rho}_{A}(\Delta)} = \frac{P^{\rho}_{A}(\{\af\})}{P^{\rho}_{A}(\Delta)},
\end{align}
where $P^{\rho}_{A} \equiv P_{\rho}( \ \cdot \ ,A)$. Analogously, the probability of $\af'$, under the evidence that $\Delta$ has occurred, is
\begin{align}
    P^{\rho}_{A}(\{\af'\} \vert \Delta) &= \frac{P^{\rho}_{A}(\{\af\})}{P^{\rho}_{A}(\Delta)}.
\end{align}
Note that subjective events are purely psychological entities and, being so, they are not subjected to the physical laws governing the physical system in question; this is why we can assume that they respect standard conditional probability. Also, the very definition of $P_{\rho}( \ \cdot \ ,A)$ \cite{KS1967, landsman2017foundations} presupposes that subjective events respect standard probability theory --- and, in particular, the standard conditional probability ---, given that $P_{\rho}( \ \cdot \ ,A)$ is a standard probability measure in a standard measurable space, where the standard rules of probability theory are valid, so our assumption is precisely the assumption behind the Born rule for infinite-dimensional quantum systems \cite{landsman2017foundations, hall2013quantum}. Since $P^{\rho}_{A}(\{\af\}\vert \Delta) + P^{\rho}_{A}(\{\af'\} \vert \Delta) = P^{\rho}_{A}(\Delta \vert \Delta) = 1$, we can define
\begin{align}
    \rho^{A}_{\Delta} \doteq P^{\rho}_{A}(\{\af\}\vert \Delta) \rho^{A}_{\af} + P^{\rho}_{A}(\{\af'\} \vert \Delta)\rho^{A}_{\af'},
\end{align}
and it is easy to see that, according to the canonical view on density operators \cite{nielsen_chuang_2010}, this is the state describing the system after measurement event $(\Delta,A)$. In fact, according to this view, $\rho^{A}_{\Delta}$ describes the situation in which the system is in one of the states $\rho^{A}_{\af}, \rho^{A}_{\af'}$ with probabilities $P^{\rho}_{A}(\{\af\}\vert \Delta),P^{\rho}_{A}(\{\af'\}\vert \Delta)$ respectively \cite{nielsen_chuang_2010}, and, as we argued, this is precisely  what the experimentalist knows about the state of the system after the event $(\Delta,A)$. In the general case, where $\Delta$ is any nonempty subset of $\sigma(A)$, a completely analogous line of though leads us to the state $\rho^{A}_{\Delta} = \sum_{\af \in \Delta} P^{\rho}_{A}(\{\af\}\vert \Delta) \rho^{A}_{\af}$. Therefore, a measurement event $(\Delta,A)$ transforms a state $\rho$ into
\begin{align}
    \rho^{A}_{\Delta} &= \sum_{\af \in \Delta} P^{\rho}_{A}(\{\af\}\vert \Delta) \rho^{A}_{\af} =  \sum_{\af \in \Delta} \frac{P^{\rho}_{A}(\{\af\})}{P^{\rho}_{A}(\Delta)} \frac{E_{\af}\rho E_{\af}}{P^{\rho}_{A}(\{\af\})}
    \\
    &= \frac{1}{\tr(\rho E_{\Delta})}\sum_{\af \in \Delta}E_{\af} \rho E_{\af},
\end{align}
where $E_{\af} \equiv \chi_{\{\af\}}(A)$ and  $E_{\Delta} \equiv \chi_{\Delta}(A) =  \sum_{\af \in \Delta}E_{\af}$. Recall that the set of density operators is a convex set, so $\rho^{A}_{\Delta}$ is a density operator. Moreover, if $\Delta = \{\af\}$, $\rho^{A}_{\Delta}$ reduces to the state  $\rho^{A}_{\af}$ given by the standard collapse postulate (definition \ref{def: collapse postulate}). Therefore, the collapse postulate can be generalized as follows.
\begin{definition}[Collapse postulate including subjective events]\label{def: subjective collapse} Let $H$ be a finite-dimensional Hilbert space. When a measurement event $(\Delta, A)$ occurs, that is to say, when a measurement of $A$ yields an outcome lying in $\Delta$, the state $\rho$ of the system is updated to
\begin{align}
    \rho^{A}_{\Delta} \doteq  \sum_{\af \in \Delta} P^{\rho}_{A}(\{\af\}\vert \Delta) \rho^{A}_{\af} = \frac{1}{\tr(\rho E_{\Delta})}\sum_{\af \in \Delta}E_{\af} \rho E_{\af},
\end{align}
where $E_{\af} \equiv \chi_{\{\af\}}(A)$ and $E_{\Delta} =  \sum_{\af \in \Delta}E_{\af}$.
\end{definition}
As we mention in section \ref{sec: general collapse}, this collapse is ``subjective'', given that it depends on how much information the experimentalist extracts form the measurement procedure.

For any event $(\Delta,A)$, the mapping $T_{(\Delta,A)}: \mathcal{S}\rightarrow \mathcal{S}$ (see definition \ref{def: F-system}) is given by
\begin{align}
    T_{(\Delta,A)}(\rho) &\doteq \sum_{\af \in \Delta} P^{\rho}_{A}(\{\af\}\vert \Delta) T_{(\af,A)}(\rho)
\end{align}
if $P_{\rho}(\Delta,A) \neq 0$, and  $T_{(\Delta,A)}(\rho) \doteq 0$ otherwise.

To fix ideas, let's analyze two particular cases. To begin with, given any state $\rho$, we obtain
\begin{align}
    \rho^{A} \equiv \rho^{A}_{\sigma(A)} &= \sum_{\af \in \sigma(A)} E_{\af} \rho E_{\af},
\end{align}
which is precisely the definition of ``loss of measurement outcome'' provided by Mark Wilde in  Ref.~\cite{wilde2013quantum}. The fact that, in general, $\rho^{A}$ is different from $\rho$ shows that our knowledge $\rho$ about the state of the system is updated even if we do not extract information from this measurement procedure. This is in accordance with the idea that measurements in quantum systems are interventions that disturb the system. 
Next, let $\rho$ be a pure state $\vert \psi \rangle \langle \psi  \vert$, where $\psi \in H$ is normalized, and assume that the spectrum of $A$ is nondegenerate, that is to say, $A$ has $n$ distinct eigenvalues $\af_{1},\dots,\af_{n}$, where $n$ is the dimension of $H$. Let $\phi_{i}$ be a normalized vector in the subspace spanned by $\af_{i}$, that is, $\vert \phi_{i}\rangle \langle \phi_{i} \vert = E_{i}$, where $E_{i} \equiv \chi_{\{\af_{i}\}}(A)$. Then, for any $\Delta \equiv \{\af_{i_{k}}: k=1,\dots m\} \subset \sigma(A)$ such that $\vert \Delta \vert >1$ and $E_{\Delta} \psi \neq 0$,
\begin{align}
    T_{(\Delta,A)}(\vert \psi \rangle \langle \psi \vert) &= \frac{1}{\langle \psi \vert E_{\Delta} \psi \rangle} \sum_{k=1}^{m} \vert \langle \phi_{i_{k}} \vert \psi\rangle \vert^{2} E_{i_{k}}
    \\
    &\equiv\frac{1}{\langle \psi \vert E_{\Delta} \psi \rangle} \sum_{k=1}^{m} \vert \langle \phi_{i_{k}} \vert \psi\rangle \vert^{2} \vert \phi_{i_{k}} \rangle \langle \phi_{i_{k}} \vert. 
\end{align}
Since pure states are extreme points in the convex set of density operators, $T_{(\Delta,A)}(\vert \psi \rangle \langle \psi \vert)$ is pure if and only if $E_{\Delta}\psi$ is an eigenvector of $A$ corresponding to an eigenvalue which lies in $\Delta$. In particular, if $\Delta = \sigma(A)$,
\begin{align}
    T_{(\sigma(A),A)}(\vert \psi \rangle \langle \psi \vert) &=  \sum_{i=1}^{n} \vert \langle \phi_{i} \vert \psi\rangle \vert^{2} P_{i}
    \\
    &\equiv  \sum_{i=1}^{n} \vert \langle \phi_{i} \vert \psi\rangle \vert^{2} \vert \phi_{i} \rangle \langle \phi_{i} \vert. 
\end{align}
It shows that measuring $A$ without extracting any information from it can disturb the state of the system even if the state is pure.



\section{Functional relations and the collapse postulate}\label{sec: conflict}

In section \ref{sec: general collapse} we mentioned that, in classical mechanics, the mapping $(\Delta,A) \mapsto A^{-1}(\Delta)$, which  associates measurement events to measurable sets, canonically induces an equivalence relation in $\mathbb{E}$, and we proved that equivalent measurement events not only are equally probable w.r.t. any state, but that they also update the state of the system in precisely the same way. The fact that functions of observables in classical mechanics are logically possible observables representing post-processings, in the sense of definition \ref{def: post-processing}, is a straightforward consequence of this fact. Now, consider the case of quantum theory. Let $H$ be a finite-dimensional Hilbert space, and let $\mathbb{E}$ be the collection of all measurement events in this system. Let $\text{Pr}(H)$ be the collection of all orthogonal projections in $H$, i.e., $E \in \text{Pr}(H)$ if and only if $E$ is a selfadjoint operator satisfying $P^{2} = P$. According to the Borel functional calculus, a measurement event $(\Delta,A)$ defines an orthogonal projection $\chi_{\Delta}(A)$, and the mapping $(\Delta,A) \mapsto \chi_{\Delta}(A)$ induces an equivalence relation $\sim$ in $\mathbb{E}$, where $(\Delta,A) \sim (\widetilde{\Delta},B)$ if and only if $\chi_{\Delta}(A) = \chi_{\widetilde{\Delta}}(B)$. In particular, given any observable $A$ and any function $g(A)$ of $A$, defined according to the functional calculus, we have, for any Borel set $\Delta$, $(\Delta,g(A)) \sim (g^{-1}(\Delta),A)$, since $\chi_{\Delta}(g(A)) = \chi_{g^{-1}(\Delta)}(A)$. Equivalent events are equally probable w.r.t. any state, so, as we have already seen, equation \ref{eq: KS definition} (equivalently, item $(a)$ from definition \ref{def: post-processing}) is satisfied. Consider now the update determined by two events $A^{\Delta} \sim B^{\widetilde{\Delta}}$. According to definition \ref{def: subjective collapse}, we have, for any state $\rho$ such that $\tr(\rho \chi_{\Delta}(A)) \neq 0$ (equivalently, $\tr(\rho \chi_{\widetilde{\Delta}}(B)) \neq 0$),
\begin{align}
    T_{(\Delta,A)}(\rho) &= \frac{1}{\tr(\rho E)}\sum_{\af \in \Delta} E_{\af} \rho E_{\af},
    \\
    T_{(\widetilde{\Delta},B)}(\rho) &= \frac{1}{\tr(\rho E)}\sum_{\beta \in \widetilde{\Delta}} F_{\beta} \rho F_{\beta},
\end{align}
where $E \equiv \chi_{\Delta}(A) = \chi_{\widetilde{\Delta}}(B)$, $E_{\af} \equiv \chi_{\{\af\}}(A)$ and $F_{\beta} \equiv \chi_{\{\beta\}}(B)$. These states are not necessarily the same, which means that equivalent events do not necessarily update the state of the system in the same way. Hence, if definition \ref{def: subjective collapse} is correct, the fact that equivalent events are associated with the same projection does not guarantee that $g(A)$ is a post-processing of $A$, in the sense of definition \ref{def: post-processing}. And, as we show in the following proposition, this really isn't the case.

\begin{proposition}\label{proposition: decisive lemma}
Let $A$ be a selfadjoint operator in a finite-dimensional Hilbert space $H$, and let $g(A) \in \mathcal{B}(H)$ be a function of $A$, defined according to the functional calculus. Then, for any eigenvalue $\beta$ of $g(A)$, the following claims are equivalent.
\begin{itemize}
    \item[(a)] Measurement events $(\beta,g(A))$ and $(g^{-1}(\beta),A)$ update the state of the system in the same way, i.e.,
    \begin{align}
        T_{(\beta, g(A))} = T_{(g^{-1}(\beta),A)}.
    \end{align}
    
    \item[(b)] $f^{-1}(\beta)$ is a singleton.
\end{itemize}
\end{proposition}
\begin{proof}
    Let $A = \sum_{\af \in \sigma(A)} \af E_{\af}$ and $g(A) = \sum_{\beta \in \sigma(g(A))} \beta F_{\beta}$ be the spectral decompositions of $A$ and $B$ respectively, and write $\Delta \equiv g^{-1}(\beta)$. We know that, given any $\beta \in \sigma(g(A))$, $F_{\beta} = \sum_{\af \in \Delta} E_{\af}$. Also, it is easy to see that, for every state $\rho$ such that $\tr(\rho F_{\beta}) \neq 0$,
    \begin{align}
       T_{(\beta,g(A))}(\rho) &=\label{eq: relating updates} T_{(\Delta,A)}(\rho) + \frac{1}{\tr(\rho F_{\beta})} \sum_{\substack{(\af,\af') \in \Delta \times \Delta \\ \af' \neq \af}} E_{\af}\rho E_{\af'},
    \end{align}
    whereas $T_{(\beta,g(A))}(\rho) = 0 = T_{(\Delta,A)}(\rho)$ otherwise. If $\Delta$ is a singleton, equation \ref{eq: relating updates} reduces to $T_{(\beta,g(A))}(\rho) = T_{(\Delta,A)}(\rho) $, so item $(a)$ follows from item $(b)$. On the other hand, suppose that $\Delta$ is not a singleton, and let $\af_{0},\af_{1}$ be distinct elements of it. Let $\phi_{i}$ be an normalized eigenvector of $A$ associated with the eigenvalue $\af_{i}$, $i=0,1$, and define $\psi \doteq \frac{1}{\sqrt{2}} (\phi_{0} + \phi_{1})$. Clearly, $g(A) \psi = \beta \psi$, thus $T_{(\beta,g(A))}(\sketbra{\psi}) = \sketbra{\psi}$, whereas
    \begin{align}
        T_{(\Delta,A)}(\sketbra{\psi}) = \frac{1}{2}(\sketbra{\phi_{0}} + \sketbra{\phi_{1}}) \neq \sketbra{\psi}.
    \end{align}
    It proves that item $(a)$ cannot be satisfied when $\vert \Delta \vert > 1$, therefore $(a)$ implies $(b)$.
\end{proof}

Hence, $g(A)$ is a post-processing of $A$ via $g$, in the sense of definition \ref{def: post-processing}, if and only if the restriction of $g$ to $\sigma(A)$ is an injective function, which means that $g(A)$ and $A$ are associated with the same partition of the identity via the spectral theorem (in other words, they are associated with the same PVM). However, this is the trivial example of a function of an observable, and it excludes functional relations of undeniable physical relevance, like $A^{2}$, the  motivating example presented by Kochen and Specker \cite{KS1967}. We thus have  the following theorem about quantum theory.

\begin{theorem}\label{theorem: TTT}
The following statements about quantum theory cannot be simultaneously true.
\begin{enumerate}
    \item The standard collapse postulate (definition \ref{def: collapse postulate}) is correct.
    \item The collapse postulate including subjective events (definition \ref{def: subjective collapse}) is correct.
    \item Functions of observables in quantum theory are theoretical representations of experimental post-processings, in the sense of definition \ref{def: post-processing}.
\end{enumerate}
\end{theorem}

\section{Discussion}\label{sec: discussion}
At least one of the three statements presented in theorem \ref{theorem: TTT} must be false. In this section we analyze these statements one by one, and we discuss the consequences of denying each one of them.

\subsection{Functional relations}
If we want to keep definitions \ref{def: collapse postulate} and \ref{def: subjective collapse}, we are forced to accept that, according to the criteria established in definition \ref{def: post-processing}, a function $g(A)$ of an observable $A$ in quantum theory is not a logically possible observable representing an experimental post-processing of $A$ via $g$, as Kochen and Specker suggest \cite{KS1967}. At first glance, it may seem a plausible conclusion, since, in some cases, the analogy between functions of observables and post-processings is not immediately clear. For instance, the momentum operator in $L^{2}(\mathbb{R})$ is the differential operator $P \equiv \frac{d}{dx}$, and $P^{2}$ is $\frac{d^{2}}{dx^{2}}$. In this case, it is not obvious that the relation between $P$ and $P^{2}$ has to be analogous to the relation between the momentum observable and its square in classical mechanics. However, the spectral theorem tells us that, if $A = \int_{\sigma(A)} \af \  dE_{\af}$ is the spectral decomposition of a selfadjoint operator $A$, then, for any Borel function $g$, we have $g(A) = \int_{\sigma(A)} g(\af) \  dE_{\af}$ \cite{barata2006curso, kadison1997fundamentals}, whereas the spectral mapping theorem states that $\sigma(g(A)) = g(\sigma(A))$ \cite{kadison1997fundamentals}. In particular, we have $P^{2} = \int_{\mathbb{R}} p^{2}\ dE_{p}$, where $P = \int_{\mathbb{R}} p \ dE_{p}$ is the spectral decomposition of $P = \frac{d}{dx}$. These theorems establish a relation between $A$ and $g(A)$ --- in particular, between $P$ and $P^{2}$ --- which is similar to the connection between $A: \Lambda \rightarrow \mathbb{R}$ and $g \circ A$ in many ways. Furthermore, recall that the whole process of quantization presupposes that functions of observables play the same role in both quantum and classical mechanics, and by doing so we are led to the correct predictions of experiments. Therefore, as we see it, dismissing Kochen and Specker's view on functional relations, although not absurd, does not seem to be the most reasonable way of escaping theorem \ref{theorem: TTT}. 

Keeping Kochen and Specker's view (more precisely, statement $3$ of theorem  \ref{theorem: TTT}) leaves us with only two options: dismissing the collapse postulate (definition \ref{def: collapse postulate}) or its generalized version (definition \ref{def: subjective collapse}). Let's begin by analyzing the latter.

\subsection{The collapse postulate for subjective events}

Recall that proposition \ref{proposition: decisive lemma} follows from the fact that, being definition \ref{def: subjective collapse} as it is, equivalent measurement events do not necessarily update the state of the system in the same way. A straightforward way of forcing equivalent events to update the system equally, and consequently saving the analogy between classical and quantum functional relations, consists in neglecting definition \ref{def: subjective collapse}  and imposing that, whenever a subjective event $(\Delta,A)$ happens, the state $\rho$ of the system is updated by $E_{\Delta} \equiv \chi_{\Delta}(A)$:
\begin{align}
    \rho^{A}_{\Delta} \doteq\label{eq: new update} \frac{E_{\Delta} \rho E_{\Delta}}{\tr(E_{\Delta}\rho E_{\Delta})}.
\end{align}
As we see it, this solution conflicts with the standard view on density operators \cite{nielsen_chuang_2010}. In fact, according to this view, the state
\begin{align}
    \underline{\rho^{A}_{\Delta}} &\doteq  \sum_{\af \in \Delta} P^{\rho}_{A}(\{\af\}\vert \Delta) \rho^{A}_{\af} =  \frac{1}{\tr(\rho E_{\Delta})}\sum_{\af \in \Delta}E_{\af} \rho E_{\af},
\end{align}
where $E_{\af} \equiv \chi_{\{\af\}}(A)$, represents the situation in which the state of the system is $\rho^{A}_{\af}$ (see definition \ref{def: collapse postulate}) with respective probability $P^{\rho}_{A}(\{\af\}\vert \Delta)$ \cite{nielsen_chuang_2010}. As we argued in section \ref{sec: collapse}, this must be the state of the system after  event $(\Delta,A)$, since this subjective event happens if and only if one of the events $\{(\af,A): \af \in \Delta \}$ occurs, and since  $P^{\rho}_{A}(\{\af\}\vert \Delta)$ is the probability of $(\af,A)$ under the evidence that $\Delta$ has happened. On the other hand, equation \ref{eq: new update} implies that, for any state $\rho$ satisfying $\tr(\rho E_{\Delta}) \neq 0$,
\begin{align}
    \rho^{A}_{\Delta} &= \underline{\rho^{A}_{\Delta}} + \frac{1}{\tr(\rho E_{\Delta})} \sum_{\substack{(\af,\af') \in \Delta \times \Delta \\ \af' \neq \af}} E_{\af}\rho P_{\af'},
\end{align}
and in the proof of lemma \ref{proposition: decisive lemma} we showed, \textit{en passant}, that we have $\rho^{A}_{\Delta} = \underline{\rho^{A}_{\Delta}}$  for every state $\rho$ if and only if $\vert \Delta \vert = 1$. Hence,  equation \ref{eq: new update} is compatible with the standard view on density operators only in the trivial case where it coincides with the collapse postulate (definition \ref{def: collapse postulate}).
Furthermore, equation \ref{eq: new update} implies that, if $\Delta = \sigma(A)$,
\begin{align}
\rho^{A} &= \rho,
\end{align}
where $\rho^{A} \equiv \rho^{A}_{\sigma(A)}$. It contradicts the definition of ``loss of measurement outcome'' provided in Ref. \cite{wilde2013quantum}. Also, it is in conflict with the idea that measurements in quantum systems disturb the system, since measuring $A$ and ignoring the outcome turns out to be equivalent to doing nothing. We thus believe that equation \ref{eq: new update} has to be discarded, and we see no better option  than definition \ref{def: subjective collapse} for the subjective collapse. 

\subsection{The collapse postulate}

Lemma \ref{proposition: decisive lemma} make it clear that the distinction between $g(A)$ and a logically possible observable derived from $A$, in the sense of definition \ref{def: post-processing}, only appear if $g$ is non-injective. The reason is that, if $g^{-1}(\beta)$ is not a singleton, then $(g^{-1}(\beta),A)$ updates the state according to definition \ref{def: subjective collapse}, whereas $(\beta,g(A))$  updates it according to the collapse postulate (definition \ref{def: collapse postulate}). This conflict vanishes if we restrict the collapse postulate to nondegenerate observables, namely observables whose spectrum is nondegenerate, and apply the subjective collapse to all degenerate ones. To put it differently, the solution consists in treating any event whose projection has rank strictly greater than one as a subjective event. 
An immediate side effect of this potential solution is that the update determined by an degenerate observable depends on a particular choice of basis. In fact, let $B$ be a degenerate observable in a $n$-dimensional system. For the sake of argument, assume that only one eigenvalue of $B$ (say, $\beta^{(0)}$) is degenerate, and let $k < n$ be the dimension of the subspace spanned by $\beta^{(0)}$ --- equivalently, $k$ is the rank of the orthogonal projection $E_{\beta^{(0)}}$ onto this subspace. Enumerate the spectrum of $B$ 
in such a way that $\beta_{i} = \beta^{(0)}$ if and only if $i \leq k$. Then, in order to know how the event $(\beta^{(0)},B)$ updates the state of the system, we need to fix a basis $\{\phi_{1},\dots,\phi_{k}\}$ within the subspace generated by $\beta^{(0)}$, i.e., we need to fix a collection $E_{i} \equiv \ketbra{\phi_{i}}{\phi_{i}}$ of rank-one orthogonal projections such that $\sum_{i=1}^{k} E_{i} = E_{\beta^{(0)}}$  and $E_{i}E_{j} = \delta_{i,j} E_{i}$. By doing so, we immediately fix a basis $\{\phi_{1},\dots,\phi_{n}\}$ for $H$, where, for any $i>k$, $E_{i} \equiv \ketbra{\phi_{i}}{\phi_{i}}$ is the projection onto the subspace spanned by $\beta_{i}$. This basis induces a set of nondegenerate observables, consisting in the collection of all real linear combinations $A \equiv \sum_{i=1}^{n} \af_{i}E_{i}$ with pairwise distinct coefficients, and, according to the functional calculus, $B$ is a function of each one of these observables. 
Therefore, by fixing a basis of eigenvectors of $B$ we are fixing a measurement context for $B$, in the sense that we are choosing the nondegenerate observable $A$  that we will measure in order to obtain $B$ by means of a post-processing. From this perspective, the collapse postulate can be redefined as follows.
\begin{definition}[Context-dependent collapse]\label{def: contextual collapse} Let $A$ be a selfadjoint operator in a $n$-dimensional Hilbert space $H$, and let $\mathfrak{B} \equiv \{E_{i}\}_{i=1}^{n}$ be a \textit{measurement basis} for $A$, that is to say, $\mathfrak{B}$ is a set of rank-one orthogonal projections satisfying, for any $i,j\in \{1,\dots,n\}$, $E_{i}E_{j} = \delta_{ij}E_{i}$ and $E_{i}A = \af_{i}E_{i} = A E_{i}$, where $\sigma(A) = \{\af_{i}: i=1,\dots,n\}$. If a measurement of $A$ in the basis $\mathfrak{B}$ yields an outcome $\af$ of $A$, the state $\rho$ of the system is updated to
\begin{align}
    \rho^{(A,\mathfrak{B})}_{\af} \doteq\sum_{\substack{i=1\\\af_{i} = \af }}^{n}  \frac{E_{i}\rho E_{i}}{\tr(\rho E_{\af})} =  \frac{1}{\tr(\rho E_{\af})}\sum_{\substack{i=1\\\af_{i} = \af }}^{n} \langle \phi_{i} \vert \rho \phi_{i}\rangle  \ketbra{\phi_{i}}{\phi_{i}},
\end{align}
where $E_{\af} \equiv \chi_{\{\af\}}(A)$, $\sketbra{\phi_{i}} = E_{i}$ and $\Vert \phi_{i}\Vert = 1$.
\end{definition}
Note that $\rho^{(A,\mathfrak{B})}_{\af}$ is pure for every state $
\rho$ if and only if $\af$ is nondegenerate. This is related to the fact that a degenerate observable can always be seen as a coarse-graining of a nondegenerate one, which in turn indicates that the distinction between degenerate and nondegenerate observables is similar to the distinction between pure and mixed states. We will discuss it in more detail latter.

Definition \ref{def: contextual collapse} suggests the incorporation of measurement bases into the definition of measurement event. From now on, by a \textit{measurement event} we mean a triple $(\Delta,A,\mathfrak{B})$, where $A$ is an observable, $\Delta$ is a subset of $\sigma(A)$ and $\mathfrak{B} \equiv \{E_{i}\}_{i=1}^{n}$ is a measurement basis for $A$. If the line of thought that led us from the collapse postulate to definition \ref{def: subjective collapse} is correct, one who accepts definition \ref{def: contextual collapse} must agree that a subjective event $(\Delta,A,\mathfrak{B})$ has to update the state $\rho$ of the system in the following manner:
\begin{align}
    T_{(\Delta,A,\mathfrak{B})}(\rho) &\doteq\label{eq: subjective T} \sum_{\af \in \Delta} P^{\rho}_{A}(\{\af\}\vert \Delta) \rho^{(A,\mathfrak{B})}_{\af} = \sum_{\af \in \Delta} \frac{P^{\rho}_{A}(\{\af\}\vert \Delta)}{\tr(\rho E_{\af})}\sum_{\substack{i=1\\\af_{i} = \af }}^{n} E_{i}\rho E_{i}
    \\
    &= \sum_{\substack{i=1\\\af_{i} = \af }}^{n} \frac{E_{i} \rho E_{i}}{\tr(\rho E_{\Delta})},
\end{align}
where $E_{\Delta} \doteq \chi_{\Delta}(A)$. If $\tr(\rho E_{\Delta}) = 0$, we define $T_{(\Delta,A,\mathfrak{B})}(\rho) \doteq 0$. Clearly, for any $\af \in \sigma(A)$, $T_{(\{\af\},A,\mathfrak{B})}(\rho) = \rho^{(A,\mathfrak{B})}_{\af}$ whenever $\tr(\rho E_{\af}) \neq 0$, so equation \ref{eq: subjective T} extends definition \ref{def: contextual collapse} to all possible measurement events. 


Distinct measurement bases for an observable $A$ never commute, i.e., if $\mathfrak{B} \equiv \{E_{i}\}_{i=1}^{n}$ and $\mathfrak{C} \equiv \{F_{i}\}_{i=1}^{n}$ are  measurement bases for $A$, then we have
\begin{align}
    \forall_{i,j}: \quad E_{i}F_{j} = F_{j}E_{i}
\end{align}
if and only if $\mathfrak{B}=  \mathfrak{C}$. Therefore, if $B$ and $C$ are nondegenerate observables associated, respectively, with distinct bases $\mathfrak{B}$ and $\mathfrak{C}$ for $A$, that is to say, if $B$ is a real linear combination of $\mathfrak{B}$ with pairwise distinct coefficients, and analogously for $C$, then $[B,C] \neq 0$. We know that $A$ is a function of both $B$ an $C$, namely $A = g(B)$ and $A = h(C)$, and we show in proposition \ref{prop: final} that measuring $A$ in the basis $\mathfrak{B}$ is equivalent to measuring $B$ and post-processing the resulting value via $g$ (analogously for $\mathfrak{C}$). As we briefly mentioned above, this is why we say that definition \ref{def: contextual collapse} is context-dependent: a measurement basis $\mathfrak{B}$ for $A$ can be seen as a measurement context for $A$, insofar it defines a set $\langle\mathfrak{B}\rangle$ of commuting nondegenerate observables such that, for any $B \in \langle \mathfrak{B} \rangle$, $\exists g: A = g(B)$, and measuring $A$ in this basis consists in measuring any $B\in \langle \mathfrak{B} \rangle$ and  post-processing the resulting value via $g$, where $A = g(B)$; furthermore, all nondegenerate observables in $\langle \mathfrak{B}\rangle$ update the state of the system in precisely the same way, and, consequently, the way $A$ updates the state of the system depends solely on the measurement basis.

The collapse due the the measurement of an observable $A$ is context-dependent if and only if $A$ is degenerate, in the sense that there is more than one measurement basis for $A$ if and only if $A$ is degenerate. Equivalently, the collapse due to a measurement of $A$ depends on the context if and only if $A$ can be written as a function  $A=g(B)=h(C)$ of noncommuting observables $B$, $C$. It is well known that it is precisely these observables that obstructs the existence of valuations (definition \ref{def: valuation})  in $\mathcal{B}(H)_{\text{sa}}$\cite{isham1998topos, doring2010thing}, since they prevent the functional composition principle from being satisfied. Hence, the collapse of the state is context-dependent in definition \ref{def: contextual collapse} for the same reason why valuations on $\mathcal{B}(H)_{\text{sa}}$ are context-dependent, by which we mean that a function $V: \mathcal{B}(H)_{\text{sa}} \rightarrow \mathbb{R}$ satisfying the `value rule'  (item $a$ from definition \ref{def: valuation}) necessarily violate equality $g(V(B)) = V(A) = h(V(C))$ for some observable $A$ such that $g(B) = A = h(C)$, where $[B,C] \neq 0$.

Let $A$ be an observable, $g(A)$ be a function of $A$, and let $\boldsymbol{\Delta}_{g}$ be the partition of $\sigma(A)$ defined by $g(A)$, namely $\boldsymbol{\Delta}_{g} \equiv \{\Delta_{\beta}: \beta \in \sigma(g(\beta))\}$, where $\Delta_{\beta} \doteq g^{-1}(\beta)$. We say that $g(A)$ is a \textit{coarse-graining} of $A$ iff $\vert \Delta_{\beta} \vert >1$ for some $\beta \in \sigma(g(B))$, which is equivalent to say that the restriction of $g$ to $\sigma(A)$ is non-injective. If $B$ is a coarse-graining of $A$, i.e., if there is a non-injective function $g: \sigma(A) \rightarrow \sigma(B)$ such that $B = g(A)$, we say that $A$ is a \textit{fine-graining} of $B$. A degenerate observable $A$ is thus a coarse-graining of $B$ for any $B \in \langle\mathfrak{B}\rangle$, where $\mathfrak{B}$ is a measurement basis for $A$, and for this reason we say that a degenerate observable can always be \textit{fine-grained}. On the other hand, a nondegenerate observable $A$ can never be fine-grained, since $\vert \sigma(A)\vert = \text{dim}(H)$; in particular, if $\mathfrak{B}$ is the unique measurement basis of a nondegenerate observable $A$ and $B \in \langle \mathfrak{B} \rangle$, then the function $g: \sigma(B) \rightarrow \sigma(A)$ satisfying $A = g(B)$ is injective. We can thus say that, when we measure a degenerate observable $A$ in a basis $\mathfrak{B}$, we are not extracting as much information from the system as this basis allows us to do. In fact, $A$ is a coarse-graining of $B$ for any $B \in \langle \mathfrak{B} \rangle$, and, as we will prove in proposition \ref{prop: final}, measuring $A$ in this basis consists in measuring some $B \in \langle \mathfrak{B} \rangle$ and post-processing the resulting value via $g$, where $A = g(B)$. Hence, accessing only the outcomes of $A$ is equivalent to accessing only the partition $\boldsymbol{\Delta}_{g}$ of $\sigma(B)$, which in turn means that we have lost the capacity of distinguishing some objective $B$-events (see section \ref{sec: collapse}), which are precisely the objective events that take place when we measure $A$ through a post-processing of $B$. On the other hand, if $A$ is nondegenerate, then there is a one-to-one correspondence between $A$-events and $B$-events and, consequently, no information is lost in accessing only the outcomes of $A$. It reinforces the aforementioned  similarity between mixed states and degenerate observables.

As we see it, the multiplicity of measurement bases for a degenerate observable resembles the variety of convex decompositions of a mixed state, and the fact that a nondegenerate observable has a unique basis is comparable to the unique convex decomposition of a pure state $\rho$, namely $\rho = 1\rho + 0 \rho$. In Spekkens' contextuality \cite{SpekkensContextuality}, distinct convex combinations of a mixed state $\rho$ are associated with distinct preparation procedures for the state $\rho$ \cite{SpekkensContextuality}. In proposition \ref{prop: final} we show that distinct measurement bases for a degenerate observable $A$ are associated with distinct measurement procedures for $A$. Thus, the dependence on contexts that appears in definition \ref{def: contextual collapse} resembles  Spekkens' notion of contextuality in quantum theory. 




The mapping $(\Delta,A, \mathfrak{B}) \mapsto T_{(\Delta,A,\mathfrak{B})}$ defines an equivalence relation $\sim$ in the collection $\mathbb{E}$ of all measurement events in $H$, and, by construction, equivalent events update the state of the system in the same way. Furthermore, it is easy to see that equivalent events are associated with the same orthogonal projection, thus, in particular, they are equally likely w.r.t. any state. In fact, let $(\Delta,A, \mathfrak{A})$ and $(\widetilde{\Delta},B, \mathfrak{B})$ be equivalent events, and write 
$E_{\Delta} \equiv \chi_{\Delta}(A)$ and  $F_{\widetilde{\Delta}} \equiv \chi_{\widetilde{\Delta}}(B)$. 
Then
\begin{align}
    T_{(\Delta,A,\mathfrak{B})}(\emptyset) = \frac{E_{\Delta}}{\tr(E_{\Delta})},
    \\
    T_{(\widetilde{\Delta},B,\mathfrak{B})}(\emptyset) = \frac{F_{\widetilde{\Delta}}}{\tr(F_{\widetilde{\Delta}})},
\end{align}
where $\phi \equiv \frac{1}{n}\mathds{1}$. Since $T_{(\Delta,A,\mathfrak{B})} = T_{(\widetilde{\Delta},B,\mathfrak{B})}$, we have $\frac{E_{\Delta}}{\tr(E_{\Delta})} = \frac{F_{\widetilde{\Delta}}}{\tr(F_{\widetilde{\Delta}})}$, which in turn is equivalent to $E_{\Delta} = F_{\widetilde{\Delta}}$. It guarantees that we do not need to take contexts into account when evaluating probabilities of events: we can refer the the probability of an event $(\Delta,A)$ without specifying the measurement basis. Note that it has been  assumed in equation \ref{eq: subjective T}.

According to our previous definition of measurement event, an event $(\Delta, A)$ canonically defines a subspace $U_{(\Delta,A)}$ of $H$, namely the subspace associated with $\chi_{\Delta}(A)$. According to the current definition, a measurement event $(\Delta, A,\mathfrak{B})$ defines a basis of $U_{(\Delta,A)}$. According to the previous definition, two measurement events are equivalent if and only if they are associated with the same subspace $U_{(\Delta,A)}$; according to the current definition, two events are equivalent if and only if they are associated not only with the same subspace $U_{(\Delta,A)}$, but also with the same basis of $U_{(\Delta,A)}$. This is why, according to the current definition, equivalent events update the state of the system in the same way. 

Let $A$ be any observable, and let $g(A)$ be a function of $A$ according to the functional calculus. Then any measurement basis for $A$ is also a measurement basis for $g(A)$. Measuring $g(A)$ in a basis  for $A$ is equivalent to measuring $A$ and post-processing the resulting value via $g$, in the sense that both procedures not only satisfy equation \ref{eq: KS definition} but also update the state of the system in precisely the same way. In fact, let $\mathfrak{B} = \{E_{i}\}_{i=1}^{n}$ be a measurement basis for $A$. Recall that $\sigma(g(A)) = \{g(\af_{i}): i=1,\dots,n\}$, where $E_{i}A = \af_{i}E_{i}$ and $\sigma(A) = \{\af_{i}: i=1,\dots,n\}$. Then, given any Borel set $\Delta$,
\begin{align}
    T_{(\Delta,g(A),\mathfrak{B})}(\rho) &=  \sum_{\substack{i=1\\ g(\af_{i}) \in \Delta}}^{n}\frac{E_{i}\rho E_{i}}{\tr(\rho \chi_{\Delta}(g(A)))}  = \sum_{\substack{i=1 \\ \af_{i} \in g^{-1}(\Delta)}}^{n}\frac{E_{i}\rho E_{i}}{\tr(\rho\chi_{g^{-1}(\Delta)}(A))}
    \\
    &= T_{(g^{-1}(\Delta),A,\mathfrak{B})}(\rho),
\end{align}
It proves the following proposition.

\begin{proposition}[Post-processing]\label{prop: final}
Let $A$ be a selfadjoint operator in a $n$-dimensional Hilbert space $H$, and let $g(A)$ be a function of $A$, defined according to the functional calculus. Then $g(A)$ is a logically possible observable representing an experimental post-processing of $A$ via $g$, by which we mean that the following conditions are satisfied.
\begin{itemize}
    \item[(a)] For any state $\rho$, the probability measure defined by $B$ matches the probability measure defined by an experimental post-processing of $A$ via $g$, i.e.,
    \begin{align}
        P_{\rho}( \ \cdot \ , g(A)) = P_{\rho}(g^{-1}( \ \cdot \ ), A).
    \end{align}
    \item[(b)] For any Borel set $\Delta$ and any measurement basis $\mathfrak{B}$ for $A$, the events $(\Delta,g(A),\mathfrak{B})$ and $(g^{-1}(\Delta),A,\mathfrak{B})$ update the state of the system in the same way, that is to say,
    \begin{align}
        T_{(\Delta,g(A),\mathfrak{B})} = T_{(g^{-1}(\Delta),A, \mathfrak{B})}.
    \end{align}
\end{itemize}
\end{proposition}

Therefore, the ``context-dependent collapse postulate'' allows us to avoid theorem \ref{theorem: TTT} without rejecting Kochen and Specker's  view on functional relations. From this perspective, measuring $g(A)$ indeed consists in measuring $A$ and evaluating $g$ in the resulting value, as pointed out by Kochen and Specker \cite{KS1967}, and, in such a procedure, the state of the system is updated by $A$. An observable $B$ can eventually be a function $B=g(A)=h(C)$ of noncommuting observables $A,C$, thus we can measure $B$ by measuring $A$ and evaluating $g$ on the resulting value, or by measuring $C$ and evaluating $h$ on its resulting value. Since $[B,C] \neq 0$, these measurement procedures are distinct, and, according to definition \ref{def: contextual collapse}, distinct procedures update the state of the system in different ways. As we mentioned, this dependence on contexts is in agreement with Kochen-Specker theorem \cite{KS1967}. This is the solution for the conflict between functional relations and the collapse postulate that sounds more convincing to us.

\section*{Acknowledgments} \label{sec:acknowledgements}
I would like to thank Bárbara Amaral and Leonardo Santos for  helpful comments. This work was funded by  National Council for Scientific and Technological Development (CNPq).

\bibliographystyle{IEEEtran}
\bibliography{Bibliography}

\end{document}